\documentclass[journal]{IEEEtran}                   
\usepackage{color}
\usepackage{psfrag}
\usepackage{balance}
\usepackage[colorlinks=true, urlcolor=blue, citecolor=blue, linkcolor=blue]{hyperref}
\usepackage{mathrsfs}
\usepackage[ruled]{algorithm2e}
\usepackage{graphicx}
\usepackage[shortlabels]{enumitem}
\usepackage{amssymb,amsfonts,amsthm}
\usepackage{soul}
\usepackage{tikz}
\usetikzlibrary{calc}
\usepackage{color}
\usepackage{pgfplots}
\pgfplotsset{compat=newest}
\usetikzlibrary{plotmarks}
\usepgfplotslibrary{patchplots}
\usepackage{grffile}
\usepackage{amsmath}
\newtheorem{remark}{Remark}
\newtheorem{lemma}{Lemma}
\newtheorem{proposition}{Proposition}
\newtheorem{problem}{Problem}
\newtheorem{property}{Property}
\newtheorem{theorem}{Theorem}

\newtheorem{definition}{Definition}

\newcommand*{\QEDB}{\hfill\ensuremath{\square}}%
\usepackage{graphicx}
\usepackage{amsmath,amssymb,amsfonts,yhmath}
\usepackage{caption}
\usepackage{subcaption}
\usepackage{color}
\usetikzlibrary{arrows}

\newcommand{\pd}{\succ 0}
\newcommand{\nd}{\prec 0}
\newcommand{\psd}{\succeq 0}
\newcommand{\nsd}{\preceq 0}
\newcommand{\Hcl}{\mathcal{H}_{cl}}

\newcommand{\N}{\mathfrak{N}}

\newcommand{\Np}{\mathfrak{N}_{> 0}}

\newcommand{\R}{\mathfrak{R}}
\newcommand{\Rp}{\mathfrak{R}_{\geq 0}}

\newcommand{\bA}{\mathbb{A}}
\newcommand{\bB}{\mathbb{B}}

\newcommand{\bJ}{\mathbb{J}}
\newcommand{\bW}{\mathbb{W}}
\newcommand{\bV}{\mathbb{V}}
\newcommand{\cA}{\mathcal{A}}

\newcommand{\cC}{\mathcal{C}}
\newcommand{\cD}{\mathcal{D}}

\newcommand{\cJ}{\mathcal{J}}
\newcommand{\cK}{\mathcal{K}}
\newcommand{\cL}{\mathcal{L}}
\newcommand{\cM}{\mathcal{M}}

\newcommand{\cP}{\mathcal{P}}

\newcommand{\cX}{\mathcal{X}}

\newcommand{\xcl}{\bar{x}}
\newcommand{\dom}{\operatorname{dom}}
\newcommand{\Id}{\mathbf{I}}
\newcommand{\minimize}{\operatorname{minimize}}
\newcommand{\Tr}{\operatorname{trace}}

\newcommand{\He}{\operatorname{He}}
\newcommand{\Diag}{\operatorname{diag}}

\newcommand{\spec}{\operatorname{spec}}
\newcommand{\source}{{This is an archival version of our paper. Please cite the published version DOI:  \href{https://doi.org/10.1109/TAC.2022.3213616}{https://doi.org/10.1109/TAC.2022.3213616}}}

\usepackage{mathtools}

\makeatletter
\def\ps@IEEEtitlepagestyle{}
\title{\LARGE \bf Robust Output Feedback Control Design in the Presence of \\Sporadic Measurements\\ (Extended Version)}
\author{Roberto Merco, Francesco Ferrante, Ricardo G. Sanfelice, Pierluigi Pisu
\thanks{Roberto Merco and Pierluigi Pisu are with the Automotive Department of Clemson University, Greenville SC 29607 USA. Email: rmerco@g.clemson.edu, pisup@clemson.edu}
\thanks{Francesco Ferrante is with Department of Engineering, University of Perugia, Via G. Duranti, 67, 06125 Perugia, Italy. Email: francesco.ferrante@unipg.it}
\thanks{Ricardo G. Sanfelice is with the Electrical and Computer Engineering Department, University of California, Santa Cruz, CA 95064, USA. Email: ricardo@ucsc.edu.
}
\thanks{Research by R. Merco and P. Pisu is supported by the National Science Foundation (NSF) under grant No. CNS-1544910. Any opinions, findings and conclusions or recommendations expressed in this material are those of the authors and do not necessarily reflect the views of the National Science Foundation.} 
\thanks{Research by F. Ferrante is funded in part by ANR via project HANDY, number ANR-18-CE40-0010.}
\thanks{Research by R. G. Sanfelice partially supported by the National Science Foundation under Grant no. CNS-2111688, Grant no. ECS-1710621, Grant no. CNS-1544396, and Grant no. CNS-2039054, by the Air Force Office of Scientific Research under Grant no. FA9550-19-1-0053, Grant no. FA9550-19-1-0169, and Grant no. FA9550-20-1-0238, and by the Army Research Office under Grant no. W911NF-20-1-0253.
}}
\usepackage{fancyhdr}
\begin{document}
\maketitle
\begin{abstract}
Output feedback control design for linear time-invariant systems in the presence of sporadic measurements and exogenous perturbations is addressed. To cope with the sporadic availability of measurements of the output, a hybrid dynamic output feedback controller equipped with a holding device whose state is reset when a new measurement is available is designed. The closed-loop system, resulting from the interconnection of the controller and the plant, is augmented with a timer variable triggering the arrival of new measurements and its properties are analyzed using hybrid system tools. Building upon Lyapunov theory for hybrid systems, sufficient conditions for internal and $\mathcal{L}_2$ input-to-output stability are proposed. A linear matrix inequalities-based design methodology for the co-design of the gains of the controller and the parameters of the holding device is presented. The effectiveness of the proposed design approach is showcased in a numerical example.
\end{abstract}
\section{Introduction}
\subsection{Background}
Over the last few decades, Networked Control Systems (NCSs) have been attracting an increasing interest in the research community; see, e.g., \cite{hespanha2007survey} and the references therein. A key feature of NCSs is the capability of sharing information, such as plant measurements and control signals, through a network. Due to the network being digital and of limited bandwidth, information exchanged between the plant and the controller happens in a sporadic fashion. In this setting, the classical paradigm assuming continuously or discretely periodically-sampled data is no longer realistic. This has brought to life an entire area of research aimed at analyzing aperiodic sampled-data systems \cite{hespanha2007survey}.

Three different main approaches have been developed  in the literature for the analysis of aperiodic sampled-data control systems.
In the input-delay approach \cite{fridman2004robust}, aperiodic sampling is modeled via a time-varying input delay. A different philosophy is followed in the lifting approach \cite{mirkin2016intermittent}, where the sampled-data control problem is converted into an equivalent discrete control problem with an infinite-dimensional input space. A complete different modeling paradigm is considered in the impulsive system approach \cite{naghshtabrizi2008exponential,briat2013convex}, in which the aperiodic sampling mechanism is captured via an impulsive dynamical system.

When only output measurements of the plant are available, controller design is more challenging and one needs to rely on output feedback strategies. In this setting, two architectures are commonly considered in the literature: observer-based controllers and dynamic output feedback controllers. An observer-based controller relies on a dynamic state estimator (or observer) and on a static state-feedback control law computed using the estimated state. Most observer-based controller architectures rely on the \emph{certainty equivalence} principle. This principle consists of plugging the estimate provided by the observer into the controller as if it were the actual plant state; see \cite{hespanha1999certainty}. The \emph{certainty equivalence} principle is adopted in \cite{bauer2013decentralized}, where switched observer-based architectures are designed for a distributed NCS. In \cite{ferrante2015observer}, the authors propose a hybrid observer-based controller to stabilize an LTI plant in the presence of intermittent measurements and sporadic input access. The approach therein relies on the use of a separation principle. 
However, when it comes to enforce disturbance attenuation properties, the use of separation principles is not appropriate as it does not enable to ensure the satisfaction of a specific input-output behavior. In this setting, the use of observer-based architectures becomes less appealing.
A typical approach to overcome this drawback consists of considering a dynamic output feedback controller that does not explicitly rely on a state estimator. Indeed, it is worthwhile to notice that not all full-order compensators are observer-based controllers; see \cite{rahman2017all}.
In this setting, closed-loop performance specifications can be imposed via Lyapunov-like conditions. This is an established philosophy in $\mathcal{H}_\infty$ control; see, e.g., \cite{scherer1997multiobjective}. Output feedback controller design for sampled-data systems is addressed in \cite{fridman2005input} via time-delay tools and in \cite{da2016regional} via looped-functional techniques.  

One common feature of the approaches discussed so far is that they rely on zero-order hold (ZOH) mechanisms to generate an estimate of the plant output in between sampling times. However, as shown in \cite{karafyllis2009continuous}
more elaborated holding devices can be employed to improve robustness and obtain enlarged maximal admissible transmission intervals.
\subsection{Contribution}
In this paper, we consider closed-loop systems where the sensing path of the system is subject to sporadic communication. In this setting, we propose a methodology for the co-design of an output feedback dynamic controller along with a general holding device. 
Notice that the setup considered in this paper, known as \textit{one-channel feedback network control system} \cite{hespanha2007survey}, constitutes a relevant case study since it can capture several configurations of network control systems; see \cite[Section III.A]{hespanha2007survey}. 
The main contributions of this work are as follows:
\begin{itemize}
\item We propose a hybrid control scheme constituted by the cascade of a general holding device and a linear dynamic output feedback controller. In particular, the controller we consider is inspired by the literature of linear $\mathcal{H}_\infty$ control and is of arbitrary structure, i.e., it is not issued from an observer-based paradigm. This enables to ensure the satisfaction of $\mathcal{H}_\infty$ specifications.
The holding device is a hybrid system whose state is reset to the plant measurement whenever a new transmission occurs and generates a signal that feeds the controller. 
\item Using a hybrid system model of the closed-loop, we propose results for the simultaneous design (co-design) of the controller parameters and of the holding device dynamics. The approach we pursue relies on Lyapunov theory for hybrid systems in the framework of \cite{goebel2012hybrid}. 
\item We provide sufficient conditions in the form of matrix inequalities to ensure:
\begin{itemize}
\item zero-input exponential stability;
\item external $\mathcal{L}_2$ stability from plant perturbations to a given regulated output with prescribed $\mathcal{L}_2$-gain;
\end{itemize}
\item An algorithm based on semidefinite programming (\emph{SDP}) tools is proposed for joint design of the controller and the holding device.
\end{itemize}
This work extends our preliminary conference paper \cite{MercoACC20}. 
In \cite{MercoACC20}, only exponential stabilization is considered. 
This paper not only extends the results in \cite{MercoACC20} to $\mathcal{L}_2$ disturbance attenuation, but also provides complete proofs (no proofs are included in \cite{MercoACC20}), a less conservative and more systematic approach towards design, and a different numerical example.

Compared to existing results, the main contribution of this work consists of the co-design of an output feedback controller and a general holding device. 
The remainder of the paper is organized as follows. Section~\ref{sec:ProbStat} introduces the  problem we solve and presents the modeling of the closed-loop system.  Section~\ref{sec:StabAnaly} presents sufficient conditions for closed-loop stability. Control design issues are addressed in Section~\ref{sec:ControlDes}. Finally, in Section~\ref{sec:Example} the effectiveness of the approach is shown in a numerical example. 

\subsection{Notation}
The symbol $\Np$ denotes the set of strictly positive integers, $\N=\Np\cup\{0\}$, $\R$ is the set of real numbers, $\Rp$ is the set of nonnegative real numbers. The Euclidean space of dimension $n$ is represented by $\R^n$ and $\R^{n\times m}$ is the set of $n \times m$ real matrices. Given $A\in \R^{n \times m}$, $A^\top$ denotes the transpose of $A$ and, when $n=m$, $A^{-\top}=(A^\top)^{-1}$ when $A$ is nonsingular, $\He(A)=A+A^\top$, and $\spec(A)$ stands for the spectrum of $A$. 
The identity matrix is denoted by $\Id$.       
The symbol $\mathsf{S}_+^n$ represents the set of $n \times n$ symmetric positive definite matrices. For a symmetric matrix $A$, $A\pd$, $A\psd$, $A\nd$, and $A \nsd$ means that $A$ is, respectively, positive definite, positive semidefinite, negative definite, negative semidefinite. The symbols $\lambda_{\min}(A)$ and $\lambda_{\max}(A)$ denote respectively the smallest and the largest eigenvalue of the matrix $A$. 
In partitioned symmetric matrices, the symbol $\bullet$ represents a symmetric block. For a vector $x \in \R^n$, $\vert x \vert$ denotes its Euclidean norm. Given two vectors $x$ and $y$, we use the equivalent notation $(x, y)=[x^\top,y^\top]^\top$. Given a vector $x \in \R^n$ and a nonempty set $\cA\in \R^n$, the distance of $x$ to $\cA$ is defined as $\vert x \vert_\cA=\inf_{y\in\cA} \vert x-y \vert$. For any function $z\colon\R \rightarrow \R^n$, we denote $z(t^+)\coloneqq \textrm{lim}_{s\rightarrow t^+} z(s)$ when it exists. Solutions to hybrid systems with inputs are represented by pairs of hybrid signals (functions defined on hybrid time domains) of the type $(\phi, u)$, where $\phi$ is a hybrid arc and $u$ is a hybrid input; see \cite{sanfelice2020hybrid} for formal definitions of hybrid signals, inputs, and arcs. Given a hybrid signal $u$, $\dom_t u\coloneqq \{t\in\R_{\geq 0}\colon \exists j\in\N\,\,\mbox{s.t.}\,\,(t, j)\in\dom u\}$ and for any $s\in\dom_t u$, $j(s)\coloneqq\min\{j\in\N\colon\,\, (s, j)\in\dom u\}$.
\section{Problem Statement and Solution Outline}
\label{sec:ProbStat}
\subsection{System Description}
We consider a plant $\cP$ described by a continuous-time linear time-invariant system of the form 
\begin{equation} \label{eq:plant}
\mathcal{P}
\left\{
\begin{aligned}
&\dot{x}_p = A_p x_p + B_p u + W_p d\\
&y = C_p x_p \\
&y_o = C_{op} x_p
\end{aligned}
\right.
\end{equation}
where $x \in \R^{n_p}$ represents the state of the plant, $u \in \R^{n_u}$ the control input, $d \in \R^{n_{d}}$  is a nonmeasurable exogenous disturbance, $y \in \R^{n_{y}}$ is the measured output of the plant, and $y_o\in \R^{n_{y_o}}$ is the regulated output\footnote{For easiness of exposition, we select the regulated output to be dependent only on the plant state. On the other hand, the approach we present can be extended to more general regulated outputs.}. The constant matrices $A_p$, $B_p$, $W_p$, $C_p$, and $C_{op}$ are given and of appropriate dimensions. We study a setup in which $u$ is a continuous-time signal, whereas $y$ is measured only at some time instances $t_k$, $k \in \Np$, not known in advance. We assume that for the sequence $\{t_k\}_{k=1}^\infty$ there exist two positive real scalars $T_1 \leq T_2$ such that
\begin{equation} \label{eq:T1T2}
0\leq t_1 \leq T_2, \quad T_1 \leq t_{k+1}-t_k \leq T_2 \quad \forall k \in \Np.
\end{equation}
The lower bound on $T_1$ in condition \eqref{eq:T1T2} introduces a strictly positive minimum time in between consecutive measurements shared after the first one. As such, this avoids the existence of Zeno behavior, which are unwanted in practice. Moreover, $T_2$ defines the Maximum Allowable Transfer Interval (MATI).
For the considered setup, the problem we solve is as follows:
\begin{problem}
\label{prob:stability}
Design an output feedback controller ensuring the following properties for the closed-loop system:
\begin{enumerate}[({P}1)] \label{prob:stability:P1}
\item The set of points in which the plant and controller states are zero\footnote{The closed-loop system resulting from our approach contains additional state variables such as timers and memory states. These variables are required to remain bounded in \emph{(P\ref{prob:stability:P1})}.} is globally exponentially stable when the input $d$ is identically zero;
\item The closed-loop system is $\cL_2$ stable from the disturbance $d$ to the regulated output $y_o$ with a prescribed $\mathcal{L}_2$ gain $\gamma>0$.
\end{enumerate}
\end{problem}
\subsection{Outline of the Proposed Solution}
To solve Problem~\ref{prob:stability:P1}, we propose an output feedback controller that relies on a linear dynamic controller $\cK$ augmented with a general holding device $\cJ$. In particular, the holding device $\cJ$, which is to be designed, is used to feed the controller $\cK$ in between measurements and its state is reset to the value of the plant output any time a new measurement gets available. 

More in detail, the continuous-time dynamic controller $\cK$ we design is given by
\begin{equation} \label{eq:controller}
\cK\left\{
\begin{aligned}
&\dot{x}_c = A_c x_c + B_c \hat{y}\\
&u = C_c x_c + D_c \hat{y},
\end{aligned}\right.
\end{equation}
where $x_c \in \R^{n_c}$ is the controller state and $\hat{y}\in\R^{n_y}$ is the state of the holding device $\cJ$. 
By making use of the last received measurement of the plant output and of the controller state, the general holding device $\cJ$ generates an intersample signal that is used to feed the controller $\cK$. In particular, for a given sequence $\{t_k\}_{k=1}^\infty$ satisfying condition \eqref{eq:T1T2}, $\cJ$ is described by
\begin{equation} \label{eq:holdingDevice}
\cJ  \left\{\begin{array}{ll}
\dot{\hat{y}}(t) =H \hat{y}(t) + E x_c(t) & \forall t \not= t_k,\\
\hat{y}(t^+) = y(t)&  \forall t = t_k. \\
\end{array} \right.
\end{equation}
The operating principle of the holding device $\cJ$ is as follows. The arrival of new measurements instantaneously updates $\hat{y}$ to $y$. In between updates, $\hat{y}$ evolves according to the continuous-time dynamics in \eqref{eq:holdingDevice} and its value is used by the controller $\cK$. The matrices 
\begin{equation}
\left[\begin{array}{c|c}
A_c & B_c \\\hline C_c & D_c
\end{array}
\right], \left[
\begin{array}{c|c}
E& H
\end{array}\right]
\label{eq:contr_params}
\end{equation}
are the parameters to be designed.

\subsection{Hybrid Modeling}
The closed-loop system can be modeled as a linear system with jumps in $\hat{y}$. In particular, for all $k \in \Np$ one obtains
\begin{equation} \label{eq:impSys}
\scalebox{0.98}{$
\begin{array}{lcl}
\begin{aligned}
& \dot{x}_p= A_p x_p\! + \! B_p C_c x_c\! + \! B_p D_c \hat{y}  \! + \! W_p d \\
& \dot{x}_c = A_c x_c  \!+\! B_c \hat{y}\\
& \dot{\hat{y}}= H\hat{y}  + E x_c
\end{aligned} & 
\left. 
\begin{aligned}
\\ \\ \\
\end{aligned}\right\} &
\begin{aligned}
& \!\!\!\! \forall t \!\not= \!t_k\\
\end{aligned}\\[0.7cm]
\begin{aligned}
& x_p(t^+) = x_p(t)\\
& x_c(t^+) = x_c(t)\\
& \hat{y}(t^+) = C_p x_p(t)
\end{aligned} & 
\left. 
\begin{aligned}
\\ \\ \\
\end{aligned}\right\} &
\begin{aligned}
& \!\!\!\! \forall t \!= \!t_k \\
\end{aligned}\\[0.7cm]
y_o(t) = C_{op} x_p(t).
\end{array}$}
\end{equation}
To devise a design algorithm for the parameters of $\cK$ and $\cJ$, we model the impulsive system in \eqref{eq:impSys} into the hybrid system framework in \cite{goebel2012hybrid}. To this end, we augment the state of the closed-loop system with the auxiliary variable $\tau \in \Rp$, which is a timer that keeps track of the duration of intervals in between transmissions of new measurement data. As in \cite{ferrante2019hybrid}, to enforce \eqref{eq:T1T2}, we make $\tau$ decrease as ordinary time $t$ increases and, whenever $\tau = 0$, we reset it to any point in $[T_1, T_2]$.
Furthermore, to simplify the analysis, we consider the change of coordinates $\eta\coloneqq C_p x_p -\hat{y}$. Hence, by taking as a state $x\coloneqq(\xcl,\eta,\tau) \in \R^{n_x}$, with $n_x\coloneqq n_p+n_c+n_y+1$ and $\xcl\coloneqq (x_p, x_c)$, the
closed-loop system can be represented by the following hybrid system
\begin{equation} \label{eq:systemH}
\Hcl
\begin{aligned}
\left\lbrace
\begin{array}{ll}
\dot{x} = f(x, d) & (x, d)\in \cC\times \R^{n_{d}}, \\
x^+  \in G(x) & x\in \cD,\\
y_o = C_o \xcl,
\end{array}
\right. 
\end{aligned}
\end{equation}
where $C_o\coloneqq\left[C_{op} \quad 0 \right]$ and the flow and jump sets are defined, respectively, as
$
\cC \coloneqq \R^{n_p+n_c+n_y}\times[0,T_2]
$, 
$\cD \coloneqq \R^{n_p+n_c+n_y}\times\{0\}$. The flow map is given for all $x \in \cC, \, d \in \R^{n_{d}}$ by
\begin{equation}\label{eq:flowMapH}
f(x, d) \coloneqq ( \bA \xcl + \bB \eta + \bV d, \bJ \xcl + \mathbb{H}\eta + \bW d,-1),
\end{equation}
with $\bA \coloneqq \left[ \begin{array}{c|c}
A_p+B_p D_c C_p &B_p C_c \\
\hline
B_c C_p & A_c
\end{array}  \right]$,
$\bB  \coloneqq  -\left[  \begin{array}{c}
B_p D_c \\ 
\hline
B_c
\end{array} \right]\!\!, \\
\bV \coloneqq \left[ \begin{array}{c}
W_p \\ 
\hline
0
\end{array} \right]$, $
\bJ \coloneqq  \left[ \begin{array}{c|c}
\bJ_1&\bJ_2
\end{array}  \right]$, and
$\bW \coloneqq C_p W_p
$
where 
$$
\begin{aligned}
&\bJ_1\coloneqq C_pA_p+C_pB_pD_c C_p-H C_p,\quad \bJ_2\coloneqq C_pB_p C_c-E,\\
&\mathbb{H}\coloneqq C_pB_pD_c-H.
\end{aligned}
$$
The jump map is defined for all $x \in \cD$, as $G(x) \coloneqq \left(\xcl, 0, \left[T_1,T_2\right] \right)$. In particular, this set-valued jump map allows to capture all possible transmission intervals of length within $T_1$ and $T_2$. Specifically, the hybrid model \eqref{eq:systemH} captures any sequence satisfying \eqref{eq:T1T2}. Since we are interested in ensuring global exponential stability of the origin of the plant, our approach to solve Problem~\ref{prob:stability} consists of designing the holding device $\cJ$ and the parameters of controller $\cK$ such that without disturbances, i.e., $d \equiv 0$ the set\footnote{Notice that, by definition of system $\Hcl$ and of the set $\cA$, for all $x\in\cC$, one has  $\vert x\vert_\cA=\vert (\xcl, \eta)\vert$. In particular, this shows that global exponential stability of $\mathcal{A}$ for \eqref{eq:systemH} implies the desired stability properties.}
\begin{equation} 
\label{eq:setA}
\cA \coloneqq \{0\} \times \{0\} \times [0,T_2]\subset\R^{n_x}
\end{equation}
is exponentially stable for $\Hcl$ in \eqref{eq:systemH}. This property is characterized by the notion of $0$-input global exponential stability defined below, which is a direct adaptation of the notion of global exponential stability as defined in  \cite{teel2012lyapunov}.
\begin{definition} \label{def:expIOstability} 
($0$-input global exponential stability)
Let $\cA\subset\R^{n_x}$ be nonempty. The set $\cA$ is \emph{$0$-input globally exponentially stable} ($0$-input GES) for the hybrid system $\Hcl$ if there exist $\kappa, \lambda>0$ such that each maximal solution pair\footnote{A pair $(\phi, d)$ is maximal if its domain cannot be extended and complete if its domain is unbounded.} $(\phi, 0)$ to $\Hcl$ is complete and satisfies 
\begin{equation} 
\label{eq:NCSAnalysisdefeISS}
\vert \phi(t, j)\vert_{\cA}\leq \kappa e^{-\lambda (t+j)}\vert \phi(0,0)\vert_{\cA} \quad \forall (t, j)\in\dom\phi.
\end{equation} 
\end{definition}
\section{Stability Analysis}
\label{sec:StabAnaly}
\subsection{Lyapunov-based Sufficient Conditions}
To solve Problem \ref{prob:stability}, in this paper we consider the closed-loop system $\Hcl$ as the interconnection of the following two systems: a continuous-time system $\Sigma_{\xcl}$ given as
\begin{subequations}
\label{eq:Sigmas}
\begin{equation}
\begin{aligned}
&\Sigma_{\xcl} \left\{
\begin{aligned}
& \dot{\bar{x}} = \bA \xcl + \bB \eta + \bV d, \\
& y_o = C_o \xcl,
\end{aligned}
\right.
\end{aligned}
\end{equation}
and a hybrid system $\Sigma_{\eta}$ given by
\begin{equation} \label{eq:jumpsInterconnectionModel}
\begin{aligned}
&\Sigma_{\eta}
\begin{aligned}
\left\lbrace\!\!
\begin{array}{ll}
\left[
\begin{array}{c}
\dot{\eta}\\
\dot{\tau}
\end{array}
\right]
\!= \!
\left[
\begin{array}{c}
\mathbb{H}\eta + \bJ \xcl + \bW d\\
-1
\end{array}\right] & 
\tau \in [0,T_2], \\
\left[
\begin{array}{c}
\eta^+\\
\tau^+
\end{array}
\right]
\!\in \!
\left[
\begin{array}{c}
0\\
\left[T_1,T_2\right]
\end{array}\right], &  \tau = 0.\\
\end{array}
\right. 
\end{aligned}
\end{aligned}
\end{equation}
\end{subequations}
This equivalent representation of $\Hcl$ in \eqref{eq:systemH} can be exploited to formulate sufficient conditions for stability of the closed-loop system $\Hcl$ by employing an approach that is reminiscent of an ``input-to-state stability small gain'' philosophy. A conceptually similar approach has been pursued in \cite{carnevale2007lyapunov} to analyze networked control systems via dissipation-like inequalities.
 
To take a first step towards the solution to Problem~\ref{prob:stability:P1}, let us consider the following property:
\begin{property}
\label{prop:Lyapunov}
Let $\gamma$ be given. There exist continuously differentiable functions $\mathscr{V}_1\colon\R^{n_p+n_c}\rightarrow\R$ and 
$\mathscr{V}_2\colon\R^{n_y+1}\rightarrow\R$, positive definite functions $\rho_1\colon\R^{n_p+n_c}\rightarrow\Rp$ and $\sigma_1\colon\R^{n_y}\rightarrow\Rp$, functions $\rho_2\colon\R^{n_y}\rightarrow\R$, $\sigma_2\colon\R^{n_c+n_p}\rightarrow\R$, $\rho_3\colon\R^{n_p+n_c+n_{d}}\rightarrow\R$, $\sigma_3\colon\R^{n_{d}}\rightarrow\R$, and  positive scalars $\overline{c}_{v_1}, \overline{c}_{v_2}, \underline{c}_{v_1}$, $\underline{c}_{v_2}$, $k_{v_1}$, $k_{v_2}$ such that  
\begin{subequations} \label{eq:Pconds}
\begin{equation}
\label{eq:Sand_1}
\underline{c}_{v_1} \vert \xcl\vert^2\leq \mathscr{V}_1(\xcl)\leq \overline{c}_{v_1} \vert \xcl\vert^{2} \quad \forall \xcl\in \R^{n_p+n_c},  
\end{equation}
\begin{equation}
\label{eq:Sand_2}
\underline{c}_{v_2} \vert \eta\vert^2\leq \mathscr{V}_2(\eta, \tau)\leq \overline{c}_{v_2} \vert \eta\vert^2 \quad \forall(\eta,\tau)\in\R^{n_y+1},
\end{equation}
\begin{equation}
\label{eq:DFlux_1_Design}
\begin{aligned}
\langle\nabla \mathscr{V}_1(\xcl), \bA\xcl+\bB \eta+ &\bV d\rangle\!  \leq \!-\rho_1(\xcl)\!+\!\rho_2(\eta)\!+\!\rho_3(\xcl,d)\\
&  \forall(\xcl,\eta)\in\R^{n_p+n_c+n_y},d \in \R^{n_{d}},
\end{aligned}
\end{equation}
\begin{equation}
\label{eq:DFlux_2_Design}
\begin{aligned}
\langle\nabla \mathscr{V}_2(\eta,\tau),& (\mathbb{H}\eta\!+\!\bJ \xcl \!+\! \bW d, -1) \rangle\!  \leq \! -\sigma_1(\eta)\!+\!\sigma_2(\xcl)\!+\!\sigma_3(d)\\
& \forall(\eta,\tau, \xcl)\in\R^{n_y}\times[0, T_2]\times\R^{n_p+n_c}, d \in \R^{n_{d}}, 
\end{aligned}
\end{equation}
\begin{equation}
\label{eq:DFlux_1_Design_2}
-\rho_1(\xcl)+\sigma_2(\xcl)\leq -k_{v_1} \vert \xcl\vert^2  \quad \forall\xcl\in\R^{n_p+n_c},
\end{equation}
\begin{equation}
\label{eq:DFlux_2_Design_2}
-\sigma_1(\eta)+\rho_2(\eta)\leq -k_{v_2} \vert \eta\vert^2  \quad \forall\eta\in\R^{n_y},
\end{equation}
\begin{equation}
\label{eq:DFlux_3_Design}
\begin{aligned}
\rho_3(\xcl,d)+\sigma_3(d)\leq & -\xcl^\top C_o^\top C_o \xcl + \gamma^2 d^\top d\\
& \forall\xcl\in\R^{n_p+n_c}, d \in \R^{n_{d}},
\end{aligned}
\end{equation}
\end{subequations}
where $C_o\coloneqq\left[C_{op} \quad 0 \right]$.
\end{property}
\begin{remark}
The satisfaction of \eqref{eq:Sand_1}-\eqref{eq:DFlux_1_Design} naturally requires the stabilizability and detectability of the plant \eqref{eq:plant}.
\end{remark}
The following theorem employs Definition~\ref{def:expIOstability} and provides sufficient conditions for the solution to Problem~\ref{prob:stability}.

\begin{theorem} \label{theo:LES}
Let Property~\ref{prop:Lyapunov} hold. Then:
\begin{enumerate}[($i$)]
\item The set $\cA$ in \eqref{eq:setA} is $0$-input GES for the hybrid closed-loop system $\Hcl$; 
\item There exists $\alpha>0$ such that any solution pair $(\phi, d)$ to $\Hcl$ satisfies
\begin{equation} \label{eq:defL2Stability}
\scalebox{0.95}{$
\!\!\!\!\sqrt{\int_\mathcal{I}\vert y_o(r, j(r))\vert^2 dr}\!\leq\!\alpha |\phi(0,0)|_\mathcal{A} + \gamma \sqrt{\int_\mathcal{I}\!\vert d(r, j(r))\vert^2 dr}
$}
\end{equation} 
where $\mathcal{I}\coloneqq\dom_t \phi$. \QEDB
\end{enumerate}
\end{theorem}
The proof of Theorem~\ref{theo:LES} is given in Appendix~\ref{sec:Appendix}. 
\begin{remark}
In principle, sufficient conditions for the solution to Problem~\ref{prob:stability} could be derived by following a similar approach as in \cite{MercoLCSS}. However, because of the coupling between the states $\bar{x}$ and $\eta$, this approach leads to conditions that are difficult to handle from a numerical standpoint. This often happens in the construction of Lyapunov functions for feedback interconnections and is the key factor leading to small-gain approaches; see, e.g., \cite{peaucelle2014lmi}.     
\end{remark}
\begin{remark}
Although inputs to \eqref{eq:systemH} are represented by hybrid signals, any purely continuous-time signal $t\mapsto\hat{w}(t)$ can be converted into a hybrid signal $w$ on a given hybrid time domain $\mathcal{E}$ by defining $w(t, j)=\hat{w}(t)$  for each $(t, j)\in \mathcal{E}$.  
\end{remark}

\begin{remark}
As opposed to \cite{briat2013convex}, the stability conditions in Theorem~\ref{theo:LES} do not depend on the value of $T_1$, \textcolor{black}{which is only required to be strictly positive for Theorem~\ref{theo:LES} to hold\footnote{The proof of Theorem~\ref{theo:LES} shows that $T_1$ has an impact on the rate of exponential convergence towards the attractor $\cA$ in \eqref{eq:setA}.}}. This is due to the fact that, by construction, the Lyapunov function $\mathscr{V}$ employed in the proof of Theorem~\ref{theo:LES} does not increase at jumps \textcolor{black}{and, for any maximal solution to $\eqref{eq:systemH}$, the length of flow intervals is lower bounded by $T_1$.}  
Although this introduces some conservatism, following this approach leads to conditions that are easier to handle for controller design.
\end{remark}

With the purpose of deriving constructive design algorithms for the controller and the holding device, we perform a particular choice for the functions $\mathscr{V}_1$ and $\mathscr{V}_2$ in Property~\ref{prop:Lyapunov}. In particular, let $P_1\in\mathsf{S}_+^{n_p+n_c}$, $P_2\in\mathsf{S}_+^{n_y}$, and $\delta$ a positive real number. Inspired by \cite{ferrante2019hybrid}, we operate the following selection:
\begin{equation} \label{eq:W1&2}
\mathscr{V}_1(\xcl) \coloneqq \xcl^\top P_1 \xcl, \quad \mathscr{V}_2(\eta,\tau) \coloneqq e^{\delta \tau} \eta^\top P_2 \eta.
\end{equation}
\subsection{Quadratic Analysis Conditions}
The structure of the selected functions $\mathscr{V}_1$ and $\mathscr{V}_2$ allows one to provide sufficient conditions for stability properties required in Problem~\ref{prob:stability} in the form of matrix inequalities.  This is formalized in the result given next.
\begin{proposition} \label{proposition:LMI}
If there exist $P_1, S, R\in\mathsf{S}_+^{n_p+n_c}$, $P_2, Q, O\in\mathsf{S}_+^{n_y}$, positive real numbers $\delta,\gamma_1,\gamma_2$, and matrices $A_c\in\R^{n_c\times n_c}$, $B_c\in\R^{n_c\times n_y}$, $C_c\in\R^{n_u\times n_c}$, $D_c\in\R^{n_u\times n_y}$, $H\in\R^{n_y\times n_y}$, and $E\in\R^{n_y\times n_c}$, such 
that 
\begin{subequations}
\begin{equation} \label{eq:LMIOtherCondQ}			
Q-O\nd,
\end{equation}
\begin{equation} \label{eq:LMIOtherCondR}
R-S\nd,
\end{equation}
\begin{equation} \label{eq:LMI_M1}
\cM_1\coloneqq\left[
\begin{array}{ccc}
\He(P_1 \bA)+S+C_o^\top C_o & P_1 \bB & P_1 \bV \\
\bullet & -Q & 0 \\
\bullet & \bullet & -\gamma_1 \Id
\end{array}
\right] \nsd,
\end{equation}
\begin{equation} \label{eq:LMI_M2}
\cM_2(0)\nsd, \quad \cM_2(T_2) \nsd,
\end{equation}
\begin{equation} \label{eq:LMI_gammas}
\gamma_1+\gamma_2 \leq \gamma^2,
\end{equation}
\end{subequations}
where for all $\tau\in [0, T_2]$
\begin{equation}\label{eq:M2}
\begin{aligned}
&\cM_2(\tau) \coloneqq\scalebox{0.9}{$\left[
\begin{array}{ccc}
(\He(P_2 \mathbb{H})-\delta P_2)e^{\delta\tau} + O& P_2 \bJ e^{\delta\tau} &  P_2 \bW e^{\delta\tau} \\
\bullet & -R & 0 \\
\bullet & \bullet & -\gamma_2 \Id
\end{array}
\right].$}
\end{aligned}
\end{equation}
Then, Property~\ref{prop:Lyapunov} holds.
\end{proposition}
\begin{proof}
Let $\mathscr{V}_1$ and $\mathscr{V}_2$ be as defined in \eqref{eq:W1&2},
$\rho_1(\xcl) \coloneqq \xcl^\top S \xcl$, $\rho_2(\eta)\coloneqq\eta^\top Q \eta$, $\rho_3(\xcl,d)\coloneqq  -\xcl^\top C_o^\top C_o \xcl + \gamma_1 d^\top d$, $\sigma_1(\eta)\coloneqq\eta^\top \textcolor{black}{O}  \eta$, $\sigma_2(\xcl) \coloneqq \xcl^\top R \xcl$, $\sigma_3(d) \coloneqq  \gamma_2 d^\top d$.
By selecting $\underline{c}_{v_1}=\lambda_{\min}(P_1)$, $\overline{c}_{v_1}=\lambda_{\max}(P_1)$, $\underline{c}_{v_2}=\lambda_{\min}(P_2)$, and $\overline{c}_{v_2}=\lambda_{\max}(P_2)e^{\delta T_2}$,
conditions \eqref{eq:Sand_1} and \eqref{eq:Sand_2} are respectively satisfied.
Regarding condition \eqref{eq:DFlux_1_Design} of Property~\ref{prop:Lyapunov}, from the definition of the flow map in \eqref{eq:flowMapH}, for each $x\in \cC$, $d\in\R^{n_{d}}$, one can define 
$
\Omega_1(\xcl,\eta,d) \coloneqq\langle \nabla \mathscr{V}_1(\xcl), \bA\xcl+\bB \eta+\bV d \rangle +\xcl^\top (S+C_o^\top C_o) \xcl - \eta^\top Q \eta - \gamma_1 d^\top d = (\xcl,\eta, d)^\top \cM_1 (\xcl,\eta, d)
$,
where the symmetric matrix $\cM_1$ is given in \eqref{eq:LMI_M1}. Therefore, the satisfaction of \eqref{eq:LMI_M1} implies \eqref{eq:DFlux_1_Design}.
Concerning condition \eqref{eq:DFlux_2_Design} of Property~\ref{prop:Lyapunov}, observe that from the definition of the flow map in \eqref{eq:flowMapH}, for each $x\in \cC$, $d\in\R^{n_{d}}$, one can define 
$
\Omega_2(\xcl,\eta,\tau,d)\coloneqq\langle \nabla  \mathscr{V}_2(\xcl), (\mathbb{H}\eta+\bJ \xcl + \bW d, -1) \rangle +\eta^\top O  \eta - \xcl^\top R \xcl - \gamma_2 d^\top d = (\eta,\xcl,d)^\top \cM_2(\tau) (\eta,\xcl, d)
$,
where the symmetric matrix $\cM_2(\tau)$ is given in \eqref{eq:M2} for all $\tau \in [0,T_2]$. Furthermore, notice that it is straightforward to show that there exists $\lambda\colon [0, T_2] \mapsto [0,1]$ such that for each $\tau \in [0,T_2]$, $\cM_2(\tau)=\lambda(\tau)\cM_2(0)+(1-\lambda(\tau))\cM_2(T_2)$; see \cite{ferrante2019hybrid} for further details. Therefore, one has that the satisfaction of \eqref{eq:LMI_M2} implies $\cM_2(\tau)\nsd, \, \forall \tau\in [0,T_2]$, hence \eqref{eq:DFlux_2_Design}.
Concerning conditions \eqref{eq:DFlux_1_Design_2} and \eqref{eq:DFlux_2_Design_2}, select
$k_{w_1}=-\lambda_{\max}(R-S)$, $k_{w_2}=-\lambda_{\max}(Q-O)$
and observe that these quantities are strictly positive due to \eqref{eq:LMIOtherCondR} and \eqref{eq:LMIOtherCondQ}. Hence, one has 
$
\xcl^\top (R-S) \xcl\leq -k_{w_1} \vert \xcl\vert^2
$, 
$
\eta^\top (Q-O) \eta\leq -k_{w_2} \vert \eta\vert^2
$
which, respectively, read as \eqref{eq:DFlux_1_Design_2} and \eqref{eq:DFlux_2_Design_2}.  
To conclude, observe that, due to \eqref{eq:LMI_gammas}, for all $\xcl\in\R^{n_p+n_c}$, $d \in \R^{n_{d}}$ one gets 
$
\rho_3(\xcl, d)+\sigma_3(d)= -\xcl^\top C_o^\top C_o \xcl + (\gamma_1 + \gamma_2) d^\top d \leq -\xcl^\top C_o^\top C_o \xcl + \gamma^2 d^\top d
$
which reads as \eqref{eq:DFlux_3_Design}. This concludes the proof. 
\end{proof}
\section{Controller Design}
\label{sec:ControlDes}
\subsection{Quadratic Design Conditions}
\label{sec:ControlDesSuff}
Proposition~\ref{proposition:LMI} enables to recast the solution to  Problem~\ref{prob:stability} into the feasibility of some matrix inequalities. However, the conditions in Proposition~\ref{proposition:LMI} are nonlinear in the variables $P_1,\,P_2,\,A_c,\,B_c,\,C_c,\,D_c,\,H,\,E$, and $\delta$. As such, those conditions are difficult to exploit from a numerical standpoint to solve Problem~\ref{prob:stability}. In this section, we show that by employing a plant-order controller, i.e., $x_c\in\R^{n_p}$ and by performing a particular selection of the matrices $H$ and $E$, the conditions in Proposition~\ref{proposition:LMI} can be turned into a collection of constraints that can be efficiently handled via semidefinite programming (\emph{SDP}) tools. 

\begin{theorem} \label{th:LMILin}
Given the plant $\cP$ in \eqref{eq:plant}, and positive scalars $\delta, \gamma$, and $ T_2$, suppose there exist $P_2, O, Q\in\mathsf{S}_+^{n_y}$, $R, F, F_i\in\mathsf{S}_+^{2n_p}$,  $X,Y\in\mathsf{S}_+^{n_p}$, $K\in\R^{n_p\times n_p}$, $L\in\R^{n_p\times n_y}$, $M\in\R^{n_u\times n_p}$, $N\in\R^{n_u\times n_y}$, $J\in\R^{n_y\times n_y}$, $Z\in\R^{n_y\times n_p}$, a nonsingular matrix $V\in\R^{n_p\times n_p}$, and positive scalars $\gamma_1,\gamma_2$ such that\footnote{Theorem~\ref{th:LMILin} can be equivalently restated by removing the constraint in \eqref{eq:ConeComp}, i.e., by replacing $F_i$ with $F^{-1}$. However, this formulation of Theorem~\ref{th:LMILin} is more suitable to derive the design algorithm outlined in Section~\ref{sec:DesignLMI}.}:
\begin{subequations} \label{eq:LMILinAllConditions}	
\begin{equation} \label{eq:LMIPhiP1Phi}	
\Theta \coloneqq 
\left[
\begin{array}{cc}
Y & \Id\\
\Id & X
\end{array}
\right] \pd,
\end{equation}
\begin{equation} \label{eq:LMIOtherCondQ_Lin}			
Q-O\nd,
\end{equation}
\begin{equation} \label{eq:LMIOtherCondR_Lin}
R-F_i\nd,
\end{equation}
\begin{equation} \label{eq:ConeComp}
FF_i=I,
\end{equation}
\begin{equation} \label{eq:LMI_M1_lin}
\scalebox{1}{$\widehat{\cM}_1 \coloneqq 
\left[
\begin{array}{ccccc}
\He(\Lambda) & \Pi & \Xi & \Phi^\top & \Phi^\top C_o^\top \\
\bullet & -Q & 0 & 0 & 0 \\
\bullet & \bullet & -\gamma_1 \Id & 0 & 0\\
\bullet & \bullet & \bullet & -F & 0\\
\bullet & \bullet & \bullet & \bullet & -\Id\\
\end{array}
\right]
\nsd,$}
\end{equation}
\begin{equation} \label{eq:LMI_M2_lin1}
\widehat{\cM}_2(0) \nsd,
\end{equation}
\begin{equation} \label{eq:LMI_M2_lin2}
\widehat{\cM}_2(T_2) \nsd,
\end{equation}
\end{subequations}
\begin{equation} \label{eq:LMI_gammas_lin}
\gamma_1+\gamma_2 \leq \gamma^2,
\end{equation}
where for all $\tau\in[0, T_2]$
\begin{equation}\label{eq:M2_lin}
\begin{aligned}
&\widehat{\cM}_2(\tau) \coloneqq\left[
\begin{array}{ccc}
e^{\delta \tau}(\He(J)-\delta P_2)+O&
e^{\delta \tau} \cM_{12}
& e^{\delta \tau} P_2 \bW 
\\[2mm]
\bullet & -R & 0\\
\bullet & \bullet & -\gamma_2 \Id
\end{array}
\right],
\end{aligned}
\end{equation}
\begin{equation} \label{eq:Phi}
\Phi \coloneqq  \left[
\begin{array}{cc}
Y & \Id\\
V^\top & 0 
\end{array}
\right], \quad 
\cM_{12} \coloneqq \left[\begin{array}{c|c}
P_2 C_p A_p - J C_p  & -Z
\end{array} \right],
\end{equation}
\begin{equation} \label{eq:LamdaPi}
\begin{aligned}
&\Lambda  \coloneqq \! \left[ 
\begin{array}{cc}
A_pY+B_p M & A_p+B_p N C_p \\
K & X A_p+L C_p
\end{array}
\right],\Pi  \coloneqq - \left[
\begin{array}{c}
B_p N \\ L
\end{array}
\right],\\
&\Xi \coloneqq \left[
\begin{array}{c}
W_p \\ X W_p
\end{array}
\right].\\
\end{aligned}
\end{equation}
Then, the matrix $\Id-XY$ is nonsingular.
Let $U\in\R^{n_p\times n_p}$ be any nonsingular matrix such that 
\begin{equation} \label{eq:Identity}
XY+UV^\top = \Id.
\end{equation}
In turn, the conditions in Proposition~\ref{proposition:LMI} are satisfied. In particular, Property~\ref{prop:Lyapunov} holds and selecting the controller and holding parameters defined in \eqref{eq:contr_params}  as in \eqref{eq:designResults} (at the top of the page) solves Problem~\ref{prob:stability}.
\end{theorem}

\begin{figure*}
\begin{equation}
\label{eq:designResults} 
\scalebox{1}{
$
\begin{aligned}
&\left[\begin{array}{c|c}
A_c&B_c\\
\hline
C_c&D_c
\end{array}\right] = \left[
\begin{array}{cc}
U^{-1} & -U^{-1} X B_p\\ 0 & \Id
\end{array}\right]
\left[
\begin{array}{c|c}
K-X A_p Y & L\\ \hline M & N
\end{array}\right]
\left[
\begin{array}{cc}
V^{-\top} & 0\\ -C_p Y V^{-\top} & \Id
\end{array}\right]\\
&\left[\begin{array}{c|c}
E&H
\end{array}\right] = \left[ 
\begin{array}{c|c}
C_pB_pC_c+P_2^{-1} Z & C_pB_pD_c+P_2^{-1} J
\end{array} \right]
\end{aligned}
$}
\end{equation}
\end{figure*}
\begin{proof}
Nonsingularity of $\Id-XY$ follows from \eqref{eq:LMIPhiP1Phi}. Indeed, from \cite[Proposition 2.8.3, page 116]{bernstein2005matrix} one has
$
\det\Theta=\det (Y)\det(X-Y^{-1})
$,
which by using the symmetry of $X$ and $Y$, via some simple algebra, yields 
$
\det\Theta=\det(YX-\Id)=(-1)^{n_p}\det(\Id-XY)
$. The remainder of the proof aims at showing that the hypotheses of the theorem imply all the conditions in the Proposition~\ref{proposition:LMI}. After a preliminary step, the satisfaction of \eqref{eq:LMIOtherCondR}, \eqref{eq:LMI_M1}, and \eqref{eq:LMI_M2} is shown.

\noindent
\textbf{Preliminary step.} Next, we select $S=F^{-1}$ and
\begin{align} \label{eq:P1Definition}
&P_1=\left[
\begin{array}{cc}
X & U \\ U^\top & -V^{-1}(Y-YXY)V^{-\top}
\end{array}
\right]
\end{align}

\noindent
\textbf{Proof of $P_1\pd$.} 
Let $P_1$ be selected as in \eqref{eq:P1Definition}. Notice that $\Phi$ in \eqref{eq:Phi} is nonsingular due to $V$ being nonsingular. Using \eqref{eq:Identity}, it can be shown that $\Theta=\Phi^{\top} P_1 \Phi$. Hence, \eqref{eq:LMIPhiP1Phi} implies $P_1\pd$.\\

\noindent
\textbf{Proof of \eqref{eq:LMIOtherCondR}.} 
Combining \eqref{eq:LMIOtherCondR_Lin} and \eqref{eq:ConeComp} yields
$
R-F^{-1}\nd
$,
which reads as \eqref{eq:LMIOtherCondR} with $S=F^{-1}$.\\

\noindent
\textbf{Proof of \eqref{eq:LMI_M1}.} By following an approach similar to \cite{scherer1997multiobjective}, we show that \eqref{eq:LMI_M1_lin} is equivalent to \eqref{eq:LMI_M1} for the proposed selection of the controller parameters and of the variables $P_1$ and $S=F^{-1}$. By Schur complement, \eqref{eq:LMI_M1} is equivalent to
\begin{equation} \label{eq:shur}
\overline{\cM}_1\coloneqq\left[
\begin{array}{ccccc}
\He(P_1 \bA) & P_1 \bB & P_1 \bV & \Id & C_o^\top  \\
\bullet & -Q & 0 & 0 & 0 \\
\bullet & \bullet & -\gamma_1 \Id & 0 & 0\\
\bullet & \bullet & \bullet & -F & 0\\
\bullet & \bullet & \bullet & \bullet & -\Id\\
\end{array}
\right]\nsd.
\end{equation}
Define
\begin{equation} \label{eq:congruence}
\begin{aligned}
\widetilde{\cM}_1\coloneqq & \Diag\{\Phi^\top,\Id\} \overline{\cM}_1 \Diag\{\Phi,\Id\} \\
= & 
\scalebox{0.9}{$\left[
\begin{array}{ccccc}
\He(\Phi^\top P_1 \bA\Phi) & \Phi^\top P_1 \bB & \Phi^\top P_1 \bV & \Phi^\top & \Phi^\top C_o^\top \\
\bullet & -Q & 0 & 0 & 0 \\
\bullet & \bullet & -\gamma_1 \Id & 0 & 0\\
\bullet & \bullet & \bullet & -F & 0\\
\bullet & \bullet & \bullet & \bullet & -\Id\\
\end{array}
\right].$}
\end{aligned}
\end{equation}
Notice that $\widetilde{\cM}_1$ differs from $\widehat{\cM}_1$ in \eqref{eq:LMI_M1_lin} only in the entries $(1,1)$, $(1,2)$, $(1,3)$, and their transposed $(2,1)$ and $(3,1)$.
Before showing that $\Lambda=\Phi^\top P_1 \bA\Phi$, $\Pi=\Phi^\top P_1 \bB$, and $\Xi=\Phi^\top P_1 \bV$, we first invert the left equation in \eqref{eq:designResults} as
\begin{equation} \label{eq:designResultsInverted}
\begin{aligned}
&\left[
\begin{array}{c|c}
K-X A_p Y & L\\ \hline M & N
\end{array}\right]\\
&\quad\quad\qquad=\left[
\begin{array}{cc}
U & X B_p\\ 0 & \Id
\end{array}\right]
\left[
\begin{array}{c|c}
A_c & B_c\\ \hline C_c & D_c
\end{array}\right]
\left[
\begin{array}{cc}
V^\top & 0\\ C_p Y & \Id
\end{array}\right].
\end{aligned}
\end{equation}
Using \eqref{eq:Identity}, by straightforward calculations one can obtain:
\begin{subequations}
\begin{equation} \label{eq:PhiAPhi}
\begin{aligned}
&\Phi^\top P_1 \bA\Phi = \\
&\!\left[ \!\!
\begin{array}{c|c}
A_pY\!+\!B_p(D_cC_pY\!+\!C_cV^{\top}) & A_p\!+\!B_pD_cC_p\\
\Gamma & XA_p\!+\!(XB_pD_c\!+\!UB_c)C_p
\end{array}\!\!
\right],
\end{aligned}
\end{equation}
\begin{equation} \label{eq:PhiP1}
\Phi^\top P_1 \bB = -\left[
\begin{array}{c}
B_p D_c \\
X B_pD_c+UB_c
\end{array}
\right], \,
\Phi^\top P_1 \bV = -\left[
\begin{array}{c}
W_p \\
X W_p
\end{array}
\right],
\end{equation}
\end{subequations}
where $\Gamma\coloneqq X(A_p+B_pD_cC_p)Y+U(B_cC_pY+A_c V^{\top})+XB_pC_cV^{\top}$.
By employing \eqref{eq:designResultsInverted}, equations \eqref{eq:PhiAPhi} and \eqref{eq:PhiP1} read as, respectively, $\Lambda$, $\Pi$, and $\Xi$ in \eqref{eq:LamdaPi}. 
This shows that \eqref{eq:LMI_M1_lin} is equivalent to \eqref{eq:LMI_M1} for the proposed selection of the controller parameters and variables $P_1$ and $S$.\\
  
\textbf{Proof of \eqref{eq:LMI_M2}.} 
Setting $H=P_2^{-1} J+C_pB_pD_c$ and $E=P_2^{-1} Z+C_pB_pC_c
$ in \eqref{eq:M2} yields \eqref{eq:M2_lin}. This shows that \eqref{eq:M2_lin} is equivalent to  \eqref{eq:M2}. Hence, \eqref{eq:LMI_M2_lin1}-\eqref{eq:LMI_M2_lin2} is equivalent to \eqref{eq:LMI_M2}.

\noindent
To conclude the proof, notice that conditions \eqref{eq:LMI_gammas_lin} and \eqref{eq:LMI_gammas} coincide.
\end{proof}
\begin{remark}
The selection of the parameters $H$ and $E$ proposed in Theorem~\ref{th:LMILin} enables to decouple the holder parameters from the controller ones. This permits the use of the typical change of coordinates/congruence transformations used in output feedback controller design \cite{scherer1997multiobjective}.
\end{remark}
\begin{remark}
Theorem~\ref{th:LMILin} requires matrix $V$ to be nonsingular. Although this constraint is hard to formulate in a matrix inequalities setting, nonsingularity of $V$ can be easily enforced, e.g., by imposing $V+V^\top\pd$. Alternatively, one can leave $V$ unconstrained and, as a second step, slightly perturb it to move away from singularity.
\end{remark}
\subsection{An SDP-based Design Algorithm}
\label{sec:DesignLMI}
The conditions in Theorem~\ref{th:LMILin} are generally hard to handle from a numerical standpoint. In particular, the main sources of difficulty come from the nonlinear dependence on the scalar variable $\delta$ in \eqref{eq:LMI_M2_lin1}-\eqref{eq:LMI_M2_lin2} and on the nonconvexity of
\eqref{eq:ConeComp}. Next, we show how these two issues can be tackled via SDP tools. In particular, the key observation is that when $\delta$ is fixed, 
\eqref{eq:LMIPhiP1Phi}, \eqref{eq:LMIOtherCondQ_Lin}, \eqref{eq:LMIOtherCondR_Lin}, \eqref{eq:LMI_M1_lin}, \eqref{eq:LMI_M2_lin1}, \eqref{eq:LMI_M2_lin2}, \eqref{eq:LMI_gammas_lin} are genuine linear matrix inequalities (\emph{LMIs}). On the other hand, the constraint \eqref{eq:ConeComp} can be handled by relying on the so-called Cone Complementarity (\emph{CC}) algorithm outlined in \cite{EOA97}. The CC algorithm can be applied in our context by relaxing the \emph{nonconvex equality} constraint \eqref{eq:ConeComp} into the following \emph{convex inequality} constraint:
\begin{equation}
\label{eq:ConvexCon}
\begin{bmatrix}
F&\Id\\
\bullet&F_i
\end{bmatrix}\succeq 0.
\end{equation}
At this stage, as in \cite{EOA97}, the idea consists of ``saturating'' the constraint \eqref{eq:ConvexCon} by minimizing\footnote{\textcolor{black}{The idea behind the relaxation proposed in \cite{EOA97} is supported by the results provided in Appendix~\ref{sec:CC_support}.}} $\Tr(FF_i)$.  
Following this approach, the design of a controller solving Problem~\ref{prob:stability} can be recast as the following optimization problem:
\begin{equation} \label{eq:NCSDesignoptProb}
\begin{array}{cl}
\minimize
& \Tr(FF_i)\\
\text{subject to}
&  \eqref{eq:LMIPhiP1Phi},\eqref{eq:LMIOtherCondQ_Lin}, \eqref{eq:LMIOtherCondR_Lin}, \eqref{eq:LMI_M1_lin}, \eqref{eq:LMI_M2_lin1}, \eqref{eq:LMI_M2_lin2},\eqref{eq:LMI_gammas_lin},\eqref{eq:ConvexCon},
\end{array}
\end{equation}
which, when $\delta$ is fixed, can be efficiently solved by using the linearization scheme proposed in \cite{EOA97}. Notice that, as indicated in \cite{EOA97}, solving \eqref{eq:NCSDesignoptProb} does not automatically guarantee the satisfaction of \eqref{eq:ConeComp}, which holds if and only if $\Tr(FF_i)=n_p$. Therefore, the satisfaction of the constraint $R-F^{-1}\nd$ needs to be checked a posteriori.

Concerning the variable $\delta$, unfortunately when this is a decision variable, \eqref{eq:LMI_M2_lin2} turns out to be a nonconvex constraint. This prevents from devising a strategy to determine a feasible value of $\delta$ better than a mere line search. Nonetheless, it is worth to remark that \eqref{eq:LMI_M2_lin1} is quasi-convex, i.e., $\widehat{\mathcal{M}}_2(0)$ is affine in $(J, P_2, R, \gamma_2, Z, O)$ for fixed $\delta$ and it satisfies the monotonicity condition $\lambda\geq\mu \implies$    
  $\widehat{\mathcal{M}}_2(0\vert \lambda)-\widehat{\mathcal{M}}_2(0\vert \mu)\nsd$, 
where, with a slight abuse of notation, $\widehat{\mathcal{M}}_2(0\vert\delta)$ denotes the matrix $\widehat{\mathcal{M}}_2(0)$ for a given value of $\delta$. Hinging upon this observation, it is possible to determine a lower bound $\underline{\delta}$ on $\delta$ such that feasibility of \eqref{eq:NCSDesignoptProb} cannot be guaranteed for $\delta<\underline{\delta}$. In particular, $\underline{\delta}$  can be determined by solving the following optimization problem:
\begin{equation} \label{eq:NCSDeltaCC}\begin{array}{cl}
\minimize
& \delta\\
\text{subject to}
&\mathcal{O}(\delta)<\infty,
\end{array}
\end{equation}
where
$$
\begin{aligned}\mathcal{O}(\delta)\coloneqq\inf\left\{
\Tr(FF_i)\colon \eqref{eq:LMIPhiP1Phi},\eqref{eq:LMIOtherCondQ_Lin}\right.,\\
\left.\eqref{eq:LMIOtherCondR_Lin}, \eqref{eq:LMI_M1_lin}, \widehat{\mathcal{M}}_2(0\vert\delta)\nsd, \eqref{eq:LMI_gammas_lin}, \eqref{eq:ConvexCon}\right\}.
\end{aligned}
$$
Due to the above mentioned quasi-convexity property, \eqref{eq:NCSDeltaCC} can be solved by performing a bisection on $\delta$, while the inner optimization problem in the definition of $\mathcal{O}(\delta)$ can be solved, at each step of the bisection, by relying on the CC algorithm.  
To summarize, \textcolor{black}{by selecting a desired upper bound $\overline{\delta}>0$ on $\delta$ and a resolution $r\in (1,\infty)$ for the line search}, a solution to Problem~\ref{prob:stability} can be obtained via Algorithm~\ref{alg:design}, 
given as follows:
\LinesNumbered
\begin{algorithm}
\SetKw{KwGoTo}{go to}
\caption{Controller design for Problem~\ref{prob:stability}}\label{alg:design}
\KwIn{Plant parameters, $T_2$, $\gamma>0$, $r\in (1, \infty)$, and $\overline{\delta}>0$.}
Solve, using bisection on $\delta$ and the CC algorithm, the optimization problem \eqref{eq:NCSDeltaCC}. 

\textcolor{black}{\eIf{problem \eqref{eq:NCSDeltaCC} is feasible}{store the optimal value of $\delta$\;}{
\KwGoTo line~\ref{alg:error}}}

\Repeat{$\textcolor{black}{\delta>\overline{\delta}}$}{
Given $\delta$ from previous step, by using the CC algorithm, solve the optimization problem \eqref{eq:NCSDesignoptProb}\;
\eIf{problem \eqref{eq:NCSDesignoptProb} is feasible and $R-F^{-1}\prec 0$}{\KwRet{controller parameters \eqref{eq:designResults}}\;}{
$\delta\longleftarrow r\times \delta$\;}
}
{\KwRet{No feasible solution is found}\label{alg:error}}
\end{algorithm}
\section{Numerical Example}
\label{sec:Example}
In this section, we showcase the proposed design approach by considering, as plant, the unicycle model linearized about the origin presented in\footnote{Numerical solutions to LMIs are obtained through the solver \textit{SDPT3} \cite{tutuncu2003solvingSDPT3} and coded in Matlab$^{\tiny{\textregistered}}$ via \textit{YALMIP} \cite{lofberg2004yalmip}. Simulations of hybrid systems are performed in Matlab$^{\tiny{\textregistered}}$ via the \textit{Hybrid Equations (HyEQ) Toolbox} \cite{sanfelice2013toolbox}.} \cite{moarref2014observer}. The state of the unicycle is defined as $x_p=(x_{p1}, x_{p2}, x_{p3})$, where $x_{p1}$ and $x_{p2}$ are, respectively, the heading angle and its time derivative, and $x_{p3}$ is the distance from the line to follow. The control input $u$ denotes the torque input, the exogenous input $d$ represents a disturbance torque acting on the unicycle's actuator. We assume that the plant measured output is $y=(x_{p1}, x_{p3})$. 
The numerical values of the matrices defining the dynamics of the unicycle are as follows:
$$
\scalebox{0.9}{$
\begin{aligned}
&\left[
\begin{array}{c|c|c|c|c}
A_p & B_p & W_p & C_p^{\top} & C_{op}^{\top}
\end{array}
\right]
=\\
&\hspace{2cm}\left[
\begin{array}{c|c|c|c|c}
\begin{array}{ccc}
0 & 1 & 0\\
0 & -0.01 & 0\\
1 & 0 & 0
\end{array} 
&
\begin{array}{c}
0\\ 1 \\ 0
\end{array}
&
\begin{array}{c}
0\\ 1 \\ 0
\end{array}
&
\begin{array}{cc}
1 & 0\\ 0 & 0 \\ 0 & 1 
\end{array}
&
\begin{array}{c}
0 \\ 0 \\ 1
\end{array}
\end{array}
\right].
\end{aligned}
$}
$$
Assuming that the output $y$ is aperiodically sampled as in \eqref{eq:T1T2} with parameters $T_1=0.1$ and $T_2=1$, we design a controller that stabilizes the unicycle, while reducing the effect of the disturbance torque $d$ on $y_o=x_{p3}$. Solving\footnote{According to standard practice, to avoid the occurrence of overly fast modes in the controller dynamics, some additional constraints on the real part and damping ratio of the eigenvalues of $\mathbb{A}$ have been added in the solution to \eqref{eq:NCSDesignoptProb}.} Problem~\ref{prob:stability} via Algorithm~\ref{alg:design} with $\gamma=10$ yields $\delta=3.1611$ and the following results for the controller and holder matrices as denoted in\footnote{The parameters of Algorithm~\ref{alg:design} are selected as follows: $r=1.1$ and $\overline{\delta}=10$. The tolerance of the bisection in step 1 is $0.1$. All numerical values obtained in the example are reported in Appendix~\ref{app:num}.} \eqref{eq:contr_params}:
\begin{equation}
\label{eq:ctrlDesign}
\begin{aligned}
&\scalebox{0.9}{$
\left[
\begin{array}{c|c}
\begin{array}{ccc}
4.74 & -1.04 & -1.54\\ -106 & 16.8 & 20.1\\ 120 & -20.7 & -25.5
\end{array} 
& 
\begin{array}{cc}
-0.27 & 0.522\\ 3.06 & -5.75\\ -4.01 & 7.6
\end{array}
\\ \hline
\begin{array}{ccc}
-215& 35.7 & 43.5 \end{array} 
&
\begin{array}{cc}
6.84 & -13 
\end{array}
\end{array}\right]$},\\
&\scalebox{0.9}{$\left[ 
\begin{array}{c|c}
\begin{array}{ccc}
-0.0634 & 0.889 & -0.959\\ 0.00323 & -0.0103 & 0.00532
\end{array}
& 
\begin{array}{cc}
-0.0787 & 0.121\\ 0.971 & -0.0211
\end{array}
\end{array} \right].$}
\end{aligned}
\end{equation}
It is interesting to notice that $\spec(H)=\{-0.3935, 0.2937\}$, i.e., the dynamics of the holding device are exponentially unstable. To showcase the performance of the designed controller, in \figurename~\ref{fig:cl_state} we show numerical solutions of the closed-loop system\footnote{In this simulation, transmission intervals are selected between $T_1$ and $T_2$ accordingly to a sinusoidal law with frequency $10.5$.}   \eqref{eq:systemH} without disturbance from the initial condition $x_p(0, 0)=(0.8, 0.1, -0.52)$, $x_c(0, 0)=0$, $\hat{y}(0, 0)=0$, $\tau(0, 0)=T_2$.
 \begin{figure}[thpb]
\centering
\psfrag{t}[][][1]{$t$}
\psfrag{xp}[][][1]{$x_p$}
\psfrag{yh}[][][1]{$\eta$}
\psfrag{u}[][][1]{$u$}
\includegraphics[width=0.95\columnwidth,trim=0.5cm 0.5cm .5cm .5cm, clip]{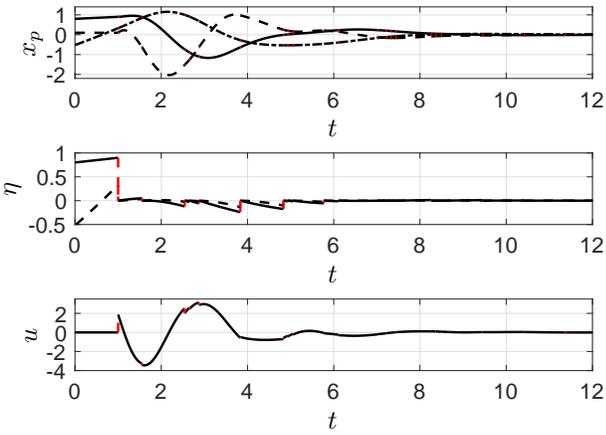}
\caption{Evolution of $x_p$, $\eta$, and $u$ with zero disturbance. Solid, dashed, and dotted lines indicate, respectively, the first, second, and third component of each state.}
\label{fig:cl_state}
\end{figure}
In \figurename~\ref{fig:cl_state_pert}, we report the response of the regulated output $y_o$ from zero initial conditions for the open and closed-loop systems and of the control input to an energy-bounded disturbance. As expected, the proposed controller is effective in reducing  the effect of the disturbance on the output, which for the open-loop system diverges. Concerning the impact of $T_2$ on the smallest achievable gain $\gamma$, numerical tests show that, as long as $\gamma$ is large enough ($\gamma\geq 20$), Algorithm~\ref{alg:design} returns feasibility for $T_2$ up to $1.6$. 
 \begin{figure}[thpb]
\centering
\psfrag{t}[][][1]{$t$}
\psfrag{i}[][][1]{$d, y_o$}
\psfrag{u}[][][1]{$u$}
\includegraphics[width=0.95\columnwidth, trim=0.5cm 0.2cm .5cm 0.2cm, clip]{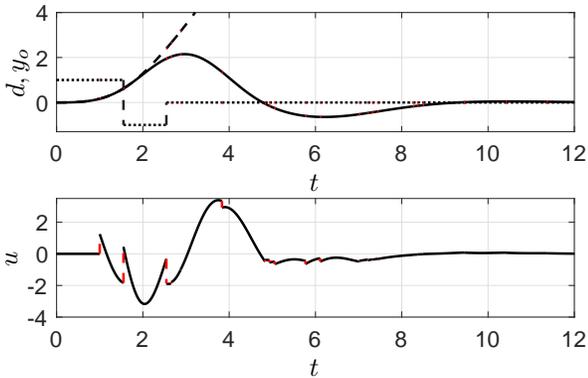}
\caption{Top picture: evolution of $y_o$ for the open-loop (dashed-line) and the closed-loop systems (solid line) in response to the disturbance $d$ (dotted-line) from zero initial conditions. Bottom picture: evolution of the control input.}
\label{fig:cl_state_pert}
\end{figure}
To gauge the impact of the parameter $T_2$ on the holding devices parameters, in \figurename~\ref{fig:norms_HE} we report the norms of $E$ and $H$ versus $T_2$ for $\gamma=10$. The picture shows that increasing $T_2$ leads to a decrease of the norm of the gain $H$, i.e., the holding device tends to ``forget'' measurements faster. On the other hand, the norm of gain $E$ increases as $T_2$ increases.
 \begin{figure}
\centering
\psfrag{T2}[][][1]{$T_2$}
\includegraphics[width=0.95\columnwidth, trim=0.5cm 0.5cm .5cm 0.5cm, clip]{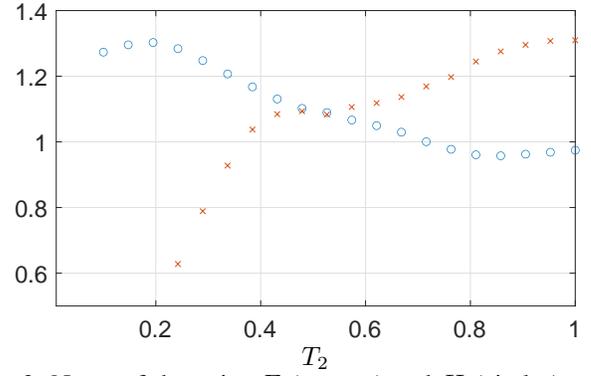}
\caption{Norm of the gains $E$ (crosses) and $H$ (circles) versus $T_2$ with $\gamma=10$.}
\label{fig:norms_HE}
\end{figure}
To further understand the effect of $T_2$, in \figurename~\ref{fig:norms_HE} we represent the maximum real part of $\spec(H)$. The picture shows a very interesting aspect, i.e., increasing $T_2$ leads to unstable dynamics for the holding device.
 \begin{figure}
\centering
\psfrag{lH}[][][1]{}
\psfrag{T2}[][][1]{$T_2$}
\includegraphics[width=0.95\columnwidth, trim=0.5cm 0.5cm .5cm 0.5cm, clip]{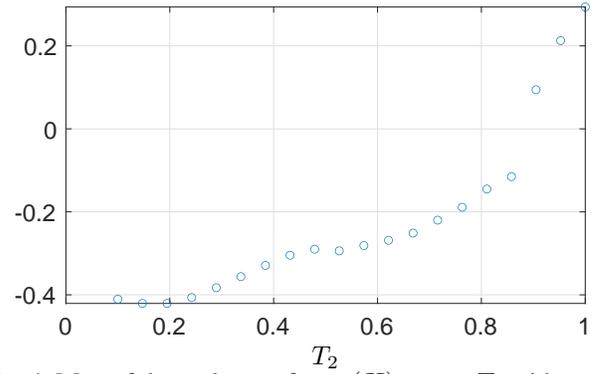}
\caption{Max of the real part of $\spec(H)$ versus $T_2$ with $\gamma=10$.}
\label{fig:norms_HE}
\end{figure}
To show the benefit of the proposed generalized holding device, in  \figurename~\ref{fig:unstable} we report  a simulation of the closed-loop system obtained when the generalized holding device is replaced by a zero-order-holder device. The figure clearly unveils that the closed-loop system becomes unstable. 
 \begin{figure}[thpb]
\centering
\psfrag{t}[][][1]{$t$}
\psfrag{xp}[][][1]{$x_p$}
\psfrag{yh}[][][1]{$\hat{y}$}
\psfrag{u}[][][1]{$u$}
\includegraphics[width=0.5\textwidth,trim=0.5cm 0.5cm .5cm .5cm, clip]{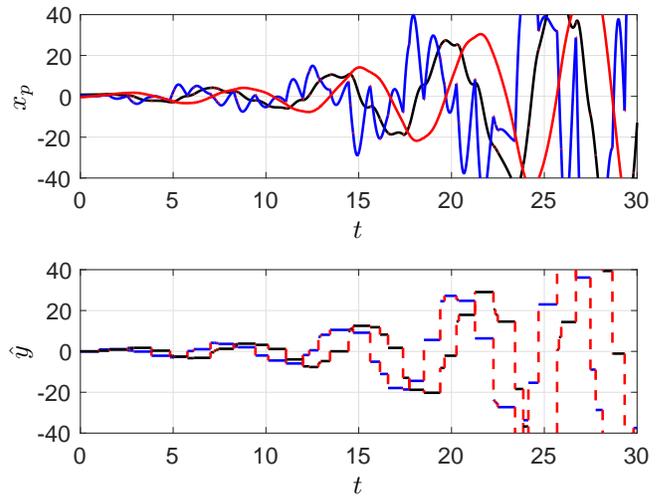}
\caption{Response with zero-order-holder implementation without disturbance from the initial condition $x_p(0, 0)=(0.8, 0.1, -0.52)$, $x_c(0, 0)=0$, $\hat{y}(0, 0)=0$, $\tau(0, 0)=T_2$.}
\label{fig:unstable}
\end{figure}
\section{Conclusion}
In this paper, we studied the problem of designing output feedback controllers for linear time-invariant systems in the presence of measurements that are available in an intermittent aperiodic fashion. In particular, the controller we propose ensures $0$-input global exponential stability and $\mathcal{L}_2$ external stability from plant disturbances to a regulated output, with prescribed $\mathcal{L}_2$-gain. A procedure based on SDP programming tools is proposed for the design of the controller. The effectiveness of the proposed approach is showcased in a numerical example. The results presented in this paper open the door to several interesting future directions. In particular, we envision to adapt the proposed controller architecture in an event-triggered control scheme. In addition, analysis of actuator saturation and robustness with respect to measurement noise for the setup studied in the paper are currently part of our research.
\bibliographystyle{plain}
\bibliography{biblio}

\appendix
\section{Appendix}
\subsection{Technical results and proofs}
\label{sec:Appendix}
{\begin{proof}[Proof of Theorem~\ref{theo:LES}
]
For all $x\in\cX$, define 
$\mathscr{V}(x)\coloneqq \mathscr{V}_1(\xcl)+\mathscr{V}_2(\eta,\tau)$. We prove ($i$) first. Select 
$
\chi_1\coloneqq\min\{\underline{c}_{v_1} ,\underline{c}_{v_2}\}
$, 
$
\chi_2\coloneqq\max\{\overline{c}_{v_1} ,\overline{c}_{v_2}\}
$.
Then, using \eqref{eq:Sand_1} and 
\eqref{eq:Sand_2} one gets
\begin{equation}
\label{eq:SandwichTotal}
\chi_1\vert x\vert^2_\mathcal{A}\leq \mathscr{V}(x)\leq \chi_2\vert x\vert^2_\mathcal{A} \qquad \forall x \in \cC \, \cup \,  \cD
\end{equation}
Moreover, by using \eqref{eq:Sand_2}, for all $x =(\xcl,\eta,\tau) \in\cD$ and $\bar{g}\in G(x)$, one has:
\begin{equation} \label{eq:Vjumps}
\mathscr{V}(\bar{g})-\mathscr{V}(x)=-\mathscr{V}_2(\eta,0)\leq 0
\end{equation}
where the first equality above follows from the fact that the value of $\xcl$ does not change at jumps.

Now observe that, from \eqref{eq:DFlux_1_Design} and \eqref{eq:DFlux_2_Design}, for all $x\in\cC$, $d\in \R^{n_{d}}$
$$
\begin{aligned}
&\langle \nabla \mathscr{V}(x), f(x,d)\rangle\\
&=\!\langle\nabla \mathscr{V}_1(\xcl), \bA\xcl\!\!+\bB \eta+\bV d \rangle\!+\!\langle\nabla \mathscr{V}_2(\eta,\tau), H \eta+\bJ \xcl\!+\!\bW d\rangle\\
&\!\!\!\overset{\eqref{eq:DFlux_1_Design},\eqref{eq:DFlux_2_Design}}{\leq} \!\!\!\!\!\!\! -\rho_1(\xcl)+\rho_2(\eta)-\sigma_1(\eta)+\sigma_2(\xcl)+\rho_3(\xcl,d)+\sigma_3(d)\\
&\!\!\!\overset{\eqref{eq:DFlux_1_Design_2},\eqref{eq:DFlux_2_Design_2},\eqref{eq:DFlux_3_Design}}{\leq} \!\!\!\!\!\!\!\!\!\!-k_{v_1} \vert \xcl\vert^2\! -\!k_{v_2} \vert \eta\vert^2 \!-\!\xcl^\top C_o^\top C_o \xcl \!+\! \gamma^2 d^\top d
\end{aligned}
$$
Then, for all $x\in\cC$, $d\in \R^{n_{d}}$
$$
\begin{aligned}
\langle \nabla \mathscr{V}(x),\! f(x,d)\rangle\!\leq\!\! -\!\underbrace{\min\{k_{v_1}\!,\! k_{v_2}\!\}}_{\chi_3}\!\vert x\vert^2_\mathcal{A}  \!-\! \xcl^\top C_o^\top\! C_o \xcl \!+\! \gamma^2 d^{\!\top} \!d
\end{aligned}
$$
Using \eqref{eq:SandwichTotal}, the above relationship yields for all $x\in\cC$, $d\in \R^{n_{d}}$
\begin{equation} \label{eq:Vflows}
\begin{aligned}
\langle \nabla \mathscr{V}(x), f(x,d)\rangle\leq -2 \lambda_t \mathscr{V}(x) -\xcl^\top C_o^\top C_o \xcl + \gamma^2 d^\top d
\end{aligned}
\end{equation}
where $\lambda_t \coloneqq \frac{\chi_3}{2\chi_2}$. We are now ready to show item $(i)$. Let $(\phi, 0)$ be a solution pair to $\mathcal{H}_{cl}$. Then, by combining \eqref{eq:Vjumps} and \eqref{eq:Vflows}, direct integration of $(t, j)\mapsto \mathscr{V}(\phi(t, j))$ yields:
$$
\mathscr{V}(\phi(t, j))\leq e^{-2\lambda_t t}\mathscr{V}(\phi(0, 0)), \quad\forall (t, j)\in\dom\phi$$
which, along with \eqref{eq:SandwichTotal}, gives:
\begin{equation} \label{eq:solRate}
\begin{aligned}
\vert\phi(t, j)\vert_\cA \leq \sqrt{\frac{\chi_2}{\chi_1}}e^{-\lambda_t t} |\phi(0,0)|_\cA \quad\forall (t, j)\in\dom\phi
\end{aligned}
\end{equation}
Finally, using Lemma~\ref{lemma:tj}, one gets that relation \eqref{eq:NCSAnalysisdefeISS} holds with any $\lambda\in\left(0,\frac{\lambda_t T_1}{1+T_1}\right]$, $\kappa=2\sqrt{\frac{\chi_2}{\chi_1}}e^\omega$, where  $\omega\geq\lambda$. Hence, since every maximal solution pair $(\phi, 0)$ to $\mathcal{H}_{cl}$ is complete, $(i)$ is established.

To establish $(ii)$, we follow a similar approach as in \cite{nevsic2013finite}. 
Let $(\phi, d)$ be a maximal solution pair to $\mathcal{H}_{cl}$. Pick any $t>0$, then thanks to \eqref{eq:Vflows} since, as shown in \eqref{eq:Vjumps}, $\mathscr{V}$ is nonincreasing at jumps, direct integration of $(t, j)\mapsto \mathscr{V}(\phi(t,j))$ yields
\begin{equation*}
\begin{split}
&\mathscr{V}(\phi(t,j))-\mathscr{V}(\phi(0,0))\leq\!\!-2\lambda_t \int_{\mathcal{I}(t)}  \mathscr{V}(\phi(s,j(s)))ds\\
&-\int_{\mathcal{I}(t)}  y_o(s,j(s))^\top  y_o(s, j(s))ds+\gamma^2 \int_{\mathcal{I}(t)} \!\vert d (s, j(s))\vert^2 ds
\end{split}
\end{equation*} 
where $\mathcal{I}(t)\coloneqq[0, t]\cap\dom_t \phi$,
which implies
\begin{equation*}
\begin{split}
\int_{\mathcal{I}(t)} y_o(s,j(s))^\top y_o(s,j(s))ds& \leq \mathscr{V}(\phi(0,0))\\
&+\gamma^2\int_{\mathcal{I}(t)} \vert d(s, j(s))\vert^2 ds
\end{split}
\end{equation*} 
Therefore, by taking the limit for $t$ approaching $\sup\dom_t\phi$, thanks to \eqref{eq:SandwichTotal}, one gets $(ii)$ with $\alpha=\chi_2$. This concludes the proof.
\end{proof}
\begin{lemma}[\cite{ferrante2019hybrid}]
\label{lemma:tj}
	Let $\lambda_t>0$, $T_1>0$,  $\lambda\in\left(0,\frac{\lambda_t T_1}{1+T_1}\right]$, and $\omega\geq\lambda$ and $(\phi, d)$ be a solution pair to $\mathcal{H}_{cl}$. Then, for all $(t, j)\in\dom\phi$
	\begin{equation}
	\label{eq:boundTJ}
	-\lambda_t t\leq \omega-\lambda(t+j)
	\end{equation}
	\QEDB
\end{lemma}
\subsection{Supporting results for the relaxation scheme in Section~\ref{sec:DesignLMI}.}
\label{sec:CC_support}}
\textcolor{black}{\begin{lemma}
Let $F, F_i\in\mathsf{S}_{+}^{n}$ and define
\begin{equation}
\label{eq:MFF}
M(F, F_i)\coloneqq\begin{bmatrix}
F&\Id\\
\Id&F_i
\end{bmatrix}\in\R^{2n\times 2n}.
\end{equation}
Suppose that 
\begin{equation}
\label{eq:Mposlemma}
M(F, F_i)\succeq 0.
\end{equation}
Then,
\begin{equation}
\label{eq:trace_bound_lemma}
\Tr(FF_i)\geq n.
\end{equation}
\end{lemma}
\begin{proof}
From \eqref{eq:Mposlemma} and $F_i$ being positive definite, direct application of Schur complement lemma yields:
$$
F-F_i^{-1}\succeq 0.
$$
Let $F_i^{\frac{1}{2}}$ be the principal square root of $F_i$. Then, from the above inequality, a simple congruence transformation gives:
$$
F_{i}^{\frac{1}{2}}FF_{i}^{\frac{1}{2}}-\Id\succeq 0
$$
which in turn, by linearity of the $\Tr$, implies
\begin{equation}
\label{eq:trace_bound_proof}
\Tr(F_{i}^{\frac{1}{2}}FF_{i}^{\frac{1}{2}}-\Id)=\Tr(F_{i}^{\frac{1}{2}}FF_{i}^{\frac{1}{2}})-n\geq 0.
\end{equation}
Finally, by using the cyclic property of the trace, one gets
$$
\Tr(F_{i}^{\frac{1}{2}}FF_{i}^{\frac{1}{2}})=\Tr(FF_i).
$$
The latter, thanks to \eqref{eq:trace_bound_proof}, gives \eqref{eq:trace_bound_lemma}. This establishes the result. 
\end{proof}
\begin{proposition}
Let $F, F_i\in\mathsf{S}_{+}^{n}$ and $M(F, F_i)$ be defined as in \eqref{eq:MFF}.
Then, the following equivalence holds
$$
\underbrace{(M(F, F_i)\succeq 0, \Tr(FF_i)=n)}_{(a)}\iff \underbrace{(M(F, F_i)\succeq 0, FF_i=\Id)}_{(b)} 
$$
\end{proposition}
\begin{proof}
The implication $(b) \implies (a)$
is trivial to show. Next we show that $(a) \implies (b)$. To this end, assume that $(a)$ holds. Then, by the Schur complement lemma and a simple congruence transformation, one gets
\begin{equation}
\label{eq:XSminus}
F_i^{\frac{1}{2}}FF_i^{\frac{1}{2}}-\Id\succeq 0
\end{equation}
where $F_i^{\frac{1}{2}}$ stands for the principal square root of $F_i$. Now observe that from the cyclic property of the trace and, by the virtue of $(a)$, it follows that
\begin{equation}
\label{eq:XSminusTrace}
\Tr(F_i^{\frac{1}{2}}FF_i^{\frac{1}{2}}-\Id)=\Tr(FF_i-\Id)=0.
\end{equation}
In particular, \eqref{eq:XSminus} and \eqref{eq:XSminusTrace} imply that 
$$
F_i^{\frac{1}{2}}FF_i^{\frac{1}{2}}-\Id=0
$$
which in turn, by simple algebraic manipulations, gives 
$$
FF_i^{-1}=\Id.
$$
Hence, item $(b)$ is established. This concludes the proof. 
\end{proof}}
\subsection{Numerical values obtained in Section~\ref{sec:Example}.}
\label{app:num}
$$
\scalebox{0.8}{$\begin{aligned}
&K=\left[\begin{array}{ccc} -0.817 & -11.2 & -0.825\\ -2.47 & 3.56 & -0.228\\ 1.34 & 11.7 & -1.48 \end{array}\right]\\
&L=\left[\begin{array}{cc} 2.43 & -7.88\\ -7.39 & 14.0\\ 4.48 & -5.93 \end{array}\right]\\
&M=\left[\begin{array}{ccc} -12.2 & -11.4 & 3.75 \end{array}\right]\\
&N=\left[\begin{array}{cc} 6.84 & -13.0 \end{array}\right]\\
&Z=\left[\begin{array}{ccc} -0.123 & 1.63 & -1.75\\ 0.185 & -2.28 & 2.4 \end{array}\right]\\
&J=\left[\begin{array}{ccc} -0.123 & 1.63 & -1.75\\ 0.185 & -2.28 & 2.4 \end{array}\right]\\
&X=\left[\begin{array}{ccc} 101& -91.5 & -47.9\\ -91.5 & 130 & 28.6\\ -47.9 & 28.6 & 30.3 \end{array}\right], Y=\left[\begin{array}{ccc} 3.87 & -2.47 & -1.33\\ -2.47 & 14.1 & -1.1\\ -1.33 & -1.1 & 1.52 \end{array}\right]\\
&V=\left[\begin{array}{ccc} 0.0526 & -1.85 & 0.491\\ 0.441 & 9.4 & -5.73\\ -0.0246 & -0.495 & 1.03\end{array}\right],\\
&U=\left[\begin{array}{ccc} -1650 & 388 & 250\\ 1480 & -390 & -173\\ 803 & -174 & -139\end{array}\right]\\
&R=\scalebox{1}{$\left[\begin{array}{cccccc} 0.349 & 0.952 & 0.0252 & 0.0164 & -0.731 & 0.925\\ 0.952 & 8.95 & -0.883 & 0.544 & -7.88 & 8.52\\ 0.0252 & -0.883 & 0.232 & -0.0533 & 0.801 & -0.924\\ 0.0164 & 0.544 & -0.0533 & 1.07 & -0.506 & 0.485\\ -0.731 & -7.88 & 0.801 & -0.506 & 7.36 & -7.16\\ 0.925 & 8.52 & -0.924 & 0.485 & -7.16 & 8.52 \end{array}\right]$}\\
&Q=\left[\begin{array}{cc} 10.3 & -16.3\\ -16.3 & 28.4 \end{array}\right], O=\left[\begin{array}{cc} 10.3 & -16.3\\ -16.3 & 28.4 \end{array}\right]
\end{aligned}$}
$$
$$
\scalebox{0.8}{$\begin{aligned}
&F=\left[\begin{array}{cccccc} 6.27 & 0.00603 & -5.2 & 0.193 & -0.084 & -1.33\\ 0.00603 & 31.1 & -12.9 & -0.193 & 16.7 & -18.4\\ -5.2 & -12.9 & 16.8 & -0.11 & -6.96 & 9.41\\ 0.193 & -0.193 & -0.11 & 0.733 & -0.0536 & 0.073\\ -0.084 & 16.7 & -6.96 & -0.0536 & 9.59 & -9.42\\ -1.33 & -18.4 & 9.41 & 0.073 & -9.42 & 11.7 \end{array}\right]\\
\end{aligned}$}
$$
$$
\scalebox{0.8}{$\begin{aligned}
&F_i=\left[\begin{array}{cccccc} 0.36 & 0.949 & 0.0144 & 0.0104 & -0.718 & 0.94\\ 0.949 & 8.95 & -0.879 & 0.546 & -7.88 & 8.52\\ 0.0144 & -0.879 & 0.242 & -0.0474 & 0.788 & -0.939\\ 0.0104 & 0.546 & -0.0474 & 1.41 & -0.514 & 0.474\\ -0.718 & -7.88 & 0.788 & -0.514 & 7.51 & -7.04\\ 0.94 & 8.52 & -0.939 & 0.474 & -7.04 & 8.64 \end{array}\right]
\end{aligned}$}
$$
$$
\scalebox{0.8}{$P_2=\left[\begin{array}{cc} 1.81 & -2.46\\ -2.46 & 9.07 \end{array}\right]$}
$$
\end{document}